\documentclass[twoside,11pt,headsepline,headinclude,footinclude,cleardoubleplain,parskip,chapterprefix,abstracton]{article}
\pdfoutput=1
\usepackage{amsmath}
\usepackage{amssymb}
\usepackage{amsthm}
\usepackage{caption}
\usepackage{makeidx}
\usepackage{setspace}
\usepackage{graphicx}
\usepackage{gensymb}
\usepackage{a4wide}
\usepackage{enumerate}
\usepackage{wrapfig}
\usepackage{tikz}
\usepackage{float}
\usepackage{listings}
\usepackage{sidecap}
\usepackage{bbm}

\providecommand{\be}{\begin{eqnarray*}}
\providecommand{\ee}{\end{eqnarray*}}

\theoremstyle{definition}
\newtheorem{thm}{Theorem}

\newtheorem{lem}[thm]{Lemma}

\pgfdeclareimage{withnucleus}{withnucleus.eps}
\pgfdeclareimage{cyto}{cyto.eps}

\makeindex

\begin{document}
\title{Mathematical Modeling of Myosin Induced Bistability of Lamellipodial Fragments}
\author{S. Hirsch\footnote{Faculty of Mathematics, University of Vienna, Oskar-Morgenstern-Platz 1, 1090 Vienna, Austria, Tel. +431427750718, Email: stefanie.hirsch@univie.ac.at}, 
A. Manhart\footnote{Faculty of Mathematics, University of Vienna, Oskar-Morgenstern-Platz 1, 1090 Vienna, Austria, Email: angelika.manhart@univie.ac.at}, 
C. Schmeiser\footnote{Faculty of Mathematics, University of Vienna, Oskar-Morgenstern-Platz 1, 1090 Vienna, Austria, Email: christian.schmeiser@univie.ac.at}}
\maketitle
\begin{abstract}
For various cell types and for lamellipodial fragments on flat surfaces, externally induced and spontaneous transitions
between symmetric nonmoving states and polarized migration have been observed. This behavior is indicative of
bistability of the cytoskeleton dynamics. In this work, the Filament Based Lamellipodium Model (FBLM), a two-dimensional,
anisotropic, two-phase continuum model for the dynamics of the actin filament network in lamellipodia, is extended by 
a new description of actin-myosin interaction. For appropriately chosen parameter values, the resulting model has bistable
dynamics with stable states showing the qualitative features observed in experiments. This is demonstrated by numerical
simulations and by an analysis of a strongly simplified version of the FBLM with rigid filaments and planar lamellipodia at
the cell front and rear.
\end{abstract}


\textbf{Acknowledgments}: Financial support by the Austrian Science Fund (FWF) through the doctoral school Dissipation and Dispersion in Nonlinear PDEs (project W1245, S.H./A.M.) as well as the Vienna Science and Technology Fund (WWTF) (project LS13/029, C.S.)
\newpage
\section{Introduction}
In a variety of physiological processes such as wound healing, immune response, or embryonic development, crawling cells play a vital role \cite{Anan}. Cell motility is the result of an interplay between protrusion at the 'front' edge of the cell (w.r.t. the direction of movement), retraction at the rear, as well as translocation of the cell body \cite{SmaRe}. 
It only occurs when the cell is polarized with a front and a differently shaped rear \cite{Kozlov}. 

Both protrusion and retraction involve the so-called \textit{lamellipodium}, a thin, sheet-like structure along the perimeter of a cell, consisting of a meshwork of \textit{actin filaments}. 
F-actin is a polar dimer that forms inextensible filaments with a fast-growing plus (barbed) end and a slow-growing minus  (pointed) end \cite{Holmes}. 

The barbed ends abut on the membrane at the leading edge \cite{Mogilner} and have a high probability of polymerization (i.e. elongation of the filament by insertion of new actin monomers), whereas at the pointed ends mostly depolymerization (removal of one monomer) or disassembly of larger parts through severing of the filament occurs.  Once a balance between polymerization and depolymerization is reached, 
each incorporated monomer is being pushed back by newly added monomers. Using the filament itself as a frame of reference, this can be described as movement of monomers from the barbed end towards the pointed end, a process called \textit{treadmilling} (see \cite{MOSS} and the references therein for an overview of the involved processes and proteins). New filaments are nucleated predominantly by branching off existing filaments. The resulting meshwork is an (almost)
two-dimensional array of (almost) diagonally arranged actin filaments with decreasing density towards the cell body \cite{SmaHe,Vinzenz}. 

The lamellipodium is stabilized by the cell membrane (surrounding the entire cell \cite{MiCra,Vallotton}), adhesions to the substrate \cite{Li,Pierini}, cross-linking proteins \cite{Nakamura,Schwaiger} and myosin II filaments \cite{Svitkina}, the
latter two binding to pairs of filaments. 
Some of the long filaments from the lamellipodium extend into the region behind, where (through the contractile effect of myosin II) forces are generated which pull the lamellipodium backwards \cite{SmaRe}.

Fish epidermal keratocytes are fast-moving cells with a relatively simple shape (circular, when stationary and crescent-moon-shaped, when moving \cite{Lee}), which makes them ideal subjects for analysis. Furthermore, they exhibit a lamellipodium with a smooth edge and a fairly uniform distribution of filaments \cite{Lacayo,SmaRe,Theriot}. 
During the transition from the stationary to the moving state, the lamellipodium in the rear of the cell collapses and the \textit{rear bundle} is formed, where myosin II generates a contractile force \cite{Svitkina,Tojk,VeSviBo}.

Treatment with staurosporine (a protein kinase inhibitor) results in the formation of completely detached lamellipodial fragments, lacking a cell body, microtubules and most other cell organelles. Remarkably, these fragments can either remain stationary while adopting a circular shape, or can move on their own, adapting their appearance to the same crescent-moon shape as the keratocyte itself \cite{Kozlov,Verkhovsky} (see Figure \ref{fig:fig1}A-C). This suggests that the necessary ingredients for movement are all present in the lamellipodium (until it runs out of energy).

\begin{figure}[h!]
\caption{ A: a moving keratocyte (right) and a moving cytoplast (left), actin is labelled in green, the nucleus in blue. B and C: a moving and a stationary cytoplast (fragment), respectively. The actin network is labelled in red, myosin in green. A, B and C are reproduced from \cite{Manhart}. E:
idealization with protruding lamellipodium at the top and lamellipodium collapsed by
actin-myosin interaction at the bottom. D: model ingredients of the simplified FBLM (clockwise, starting top left):
cross-link stretching, cross-link twisting, filament-substrate adhesion, connection
between front and rear by stress fibres, membrane stretching, actin-myosin interaction.}
\includegraphics[width=\textwidth]{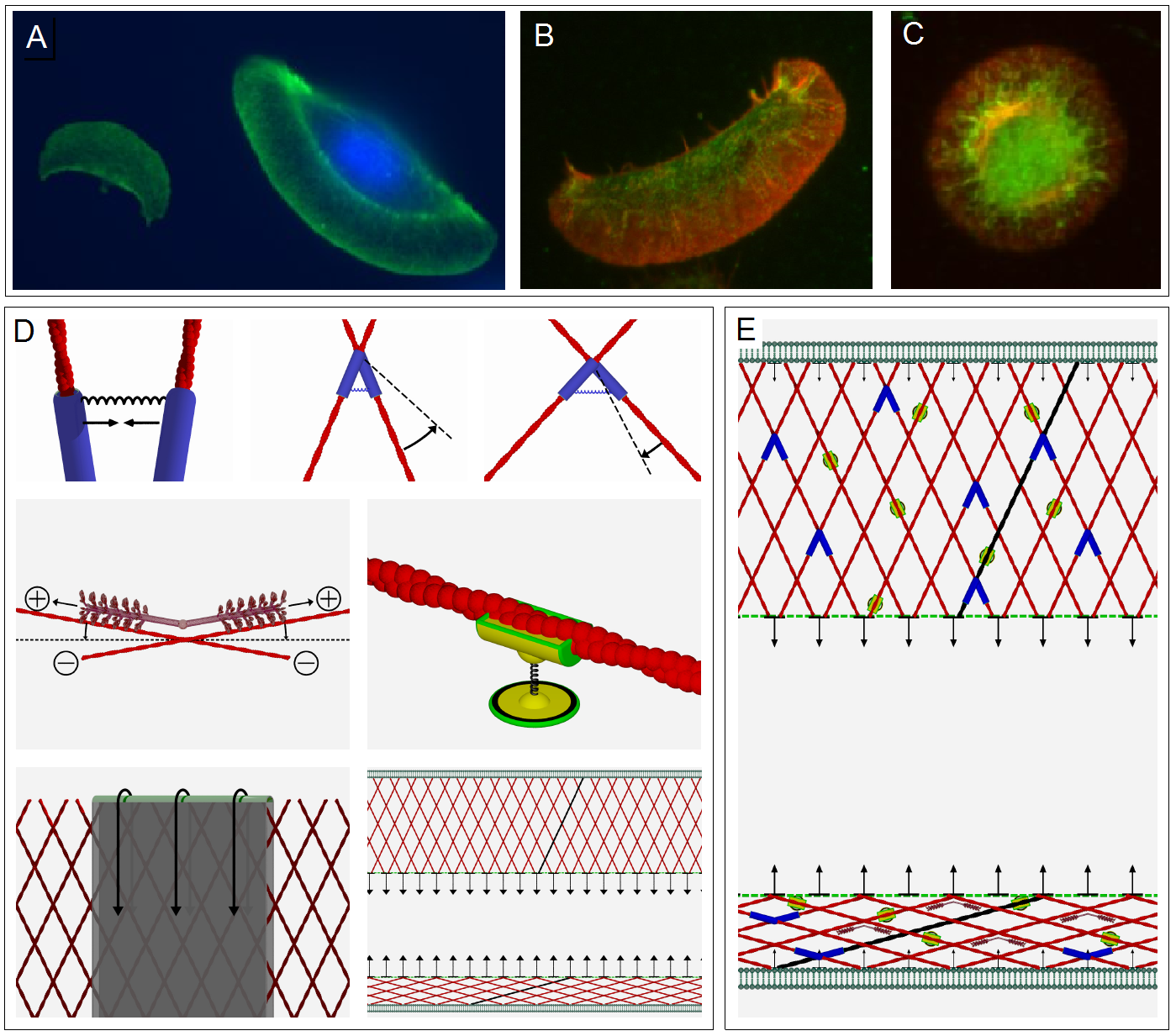}

\label{fig:fig1}
\end{figure}

Various approaches to continuum mechanical modeling of the lamellipodium exist \cite{Kruse,Rubinstein,Ziebert}.  
This work is based on the FBLM \cite{MOSS,OelSch1,OelSch2}, a two-dimensional, anisotropic, two-phase model
derived from a microscopic (i.e. individual filament based) description, accounting for most of the phenomena 
mentioned above. It describes the actin network in terms of two transversal families of locally parallel filaments,
stabilized by transient cross-links and substrate adhesions. In Section \ref{sec:add-myo} the FBLM is presented and
extended by a model for actin-myosin interaction between the two families. We assume that myosin filaments can 
connect only when the families are anti-parallel enough and they are described as transient, similar to cross-links. 
They tend to slide the two families relative to each other, and they are assumed to have a turning effect, making the two
families more anti-parallel. These properties are expected to produce the desired bistable behavior. This is 
demonstrated by numerical simulations in Section \ref{sec:sim}, which indicate the existence of two stable states,
a rotationally symmetric nonmoving state and a polarized state, where the cell moves. The moving state is characterized 
by a more anti-parallel network in the rear of the cell, where actin-myosin interaction is active. Complete collapse
of the network and consequential generation of a rear bundle are avoided, since the FBLM is (so far) unable to deal
with such topological changes.

The occurrence of bistability is also proven analytically for a strongly simplified model. In Section \ref{sec:rigid} the 
complexity of the model is reduced in a first step by assuming rigid filaments. Then a planar, translationally invariant
lamellipodium is considered in Section \ref{sec:planar}, which reduces the model to a system of three ordinary differential equations. Here we also neglect the effects of branching and capping, assumed to be in equilibrium, as well as filament severing
within the modelled part of the lamellipodium, implying a constant actin density there.
Bistability is obtained for this model in Section \ref{sec:steady}. Finally, in Section \ref{sec:coupling}
a cell (fragment) is replaced by a pair of connected back-to-back planar lamellipodia, and the existence of stable
stationary (symmetric) as well as moving (polarized) states is proven. The same bistable behavior is observed in the simulations of the full model in Section \ref{sec:sim}.\\

Figure \ref{fig:fig1} depicts the main components of the simplified version of the FBLM (D and E) together with one keratocyte and three fragments (A-C). The crescent-moon shaped cells and cell fragments are moving, whereas the circularly shaped fragment remains stationary. One can also observe that in moving fragments myosin can predominantely be found at the cell rear. In Figure \ref{fig:fig1}E, the idealized model obtained in Sections \ref{sec:rigid}-\ref{sec:coupling} is illustrated. It can be interpreted as description of lamellipodial sections at
the front and at the rear of the cell. The main model ingredients are depicted in \ref{fig:fig1}D: diagonally arranged filaments (red), 
the cell membrane (green, with arrows indicating the force acting on the barbed ends due to membrane tension), 
cross-links (blue, producing friction between the filament families and a turning force trying to establish an equilibrium
angle), adhesions (yellow, producing friction relative to the substrate), myosin filaments (pink, trying to slide the filament
families and to make them anti-parallel), and the inward pulling forces due to stress fibers in the interior of the cell
(dashed green line and arrows).

\section{Adding actin-myosin interaction to the Filament Based Lamellipodium Model (FBLM)} \label{sec:add-myo}

Our starting point is the FBLM as introduced in \cite{OelSch1} (see also \cite{MOSS}):
\begin{eqnarray} \label{FBLM}
  0&=&\mu^B \partial_s^2\left(\eta \partial_s^2 F\right) + \mu^A \eta D_t F
    - \partial_s\left(\eta \lambda_\text{inext} \partial_s F\right)\\
&&+ \widehat{\mu^S}\eta\eta^* (D_t F - D_t^* F^*) \pm \partial_s\left(\widehat{\mu^T}\eta\eta^* (\varphi-\varphi_0)\partial_s F^{\perp}\right) \,,\nonumber
\end{eqnarray}
where $F=F(\alpha,s,t)\in\mathbb{R}^2$ describes the position and deformation of actin filaments in the plane at time $t$. More precisely, the variable 
$\alpha\in A\subset\mathbb{R}$, for some interval $A$, is a filament label, and $s\in[-L(\alpha,t),0]$ denotes an arclength parameter along filaments, which means
that the constraint 
\begin{equation} \label{inextensibility}
  |\partial_s F| = 1
\end{equation}
has to be satisfied. Here $L(\alpha,t)$ is the maximal length of filaments in an infinitesimal region $d\alpha$ around $\alpha$. The filament length density
with respect to $\alpha$ and $s$ is given by $\eta(\alpha,s,t)$, which will be assumed as given (see \cite{MOSS} for a dynamic model incorporating
polymerization, depolymerization, nucleation, and branching effects). The value $s=0$ corresponds to the so called \emph{barbed ends} of the polar
filaments, abutting the leading edge of the lamellipodium. The rear boundary $s=-L(\alpha,t)$ is introduced somewhat artificially since the rear end of the
lamellipodium is typically not well defined. By polymerization with speed $v(\alpha,t)$ (also assumed as given in this work), monomers move
along filaments in the negative $s$-direction. Their speed relative to the nonmoving substrate is therefore given by $D_t F$ with the material derivative
$D_t = \partial_t - v\partial_s$.

The terms in the first line of \eqref{FBLM} correspond to the filaments' resistance against bending with stiffness parameter $\mu^B$, to friction relative
to the substrate as a consequence of adhesion dynamics with adhesion coefficient $\mu^A$, and to the constraint \eqref{inextensibility} with the
Lagrange multiplier $\lambda_\text{inext}$.

The FBLM is actually a two phase model, and $F$ may stand for either of the two families $F^+$ or $F^-$. The terms in the second line of \eqref{FBLM}
describe the interaction between the two families, with the other family indicated by the superscript $*$. The interaction is the consequence of dynamic
cross-linking and leads to a friction term proportional to the relative velocity between the two families and to a turning force trying to push the angle
$\varphi$ ($\cos\varphi = \partial_s F\cdot\partial_s F^*$) between crossing filaments to its equilibrium value $\varphi_0$, corresponding to the 
equilibrium conformation of the cross-linker molecule
($F^\bot = (-F_y,F_x)$). The $*$-quantities corresponding to the other family have to be evaluated at $(\alpha^*,s^*)$, determined by the requirement
$F(\alpha,s,t)= F^*(\alpha^*,s^*,t)$. It is a basic geometric modeling assumption that the coordinate change $(\alpha,s) \leftrightarrow (\alpha^*,s^*)$
is one-to-one, wherever the two families overlap. It requires that filaments of the same family do not cross each other and that pairs of filaments of different families cross each other at
most once. Finally, the coefficients are given by
\begin{equation} \label{coeff-cl}
  \widehat{\mu^S} = \mu^S \left|\frac{\partial{\alpha}^*}{\partial s}\right| \,,\qquad 
  \widehat{\mu^T} = \mu^T \left|\frac{\partial{\alpha}^*}{\partial s}\right| \,,
\end{equation}
with constants $\mu^{S,T}$, wherever $F$ crosses another filament, and zero elsewhere. The partial derivative refers
to the coordinate transformation introduced above.

The model will be extended by the effects of myosin polymers. The basic modeling assumption is that pairs of crossing actin filaments, which lie
antiparallel enough, may be connected by a bipolar myosin filament. The modeling is similar to cross-links. However, by their motor activity, the myosin
heads have the tendency to move towards the barbed end of the actin filament the myosin filament is attached to. We assume a constant equilibrium
speed $v^M$ of this movement. Transient building and breaking of actin-myosin connections are assumed to cause a friction effect. Furthermore
the actin-myosin interaction is assumed to have a turning effect on the actin filaments, which tends to align them in the antiparallel direction. This
is similar to the turning effect of cross-links, however now with the equilibrium angle $\pi$. We also assume that myosin can only act on pairs of filaments, if they are antiparallel enough, i.e. if their angle is between some cut-off value $\overline \varphi$ and $\pi$.

The modified model has the form
\begin{eqnarray} 
  0&=&\mu^B \partial_s^2\left(\eta \partial_s^2 F\right) + \mu^A \eta D_t F
    - \partial_s\left(\eta \lambda_\text{inext} \partial_s F\right) \nonumber\\
&&+ \widehat{\mu^S}\eta\eta^* (D_t F - D_t^* F^*) \pm \partial_s\left(\widehat{\mu^T}\eta\eta^* (\varphi-\varphi_0)\partial_s F^{\perp}\right)
     \label{FBLMplus}\\
&& + \widehat{\mu^{SM}}\eta\eta^* (D_t F - D_t^* F^*+ v^M(\partial_s F - \partial_s F^*)) \pm \partial_s\left(\widehat{\mu^{TM}}\eta\eta^* (\varphi-\pi)\partial_s F^{\perp}\right)\,,\nonumber
\end{eqnarray}
with 
\begin{equation} \label{coeff-myo}
  \widehat{\mu^{SM}} = \mu^{SM}(\varphi) \left|\frac{\partial{\alpha}^*}{\partial s}\right| \,,\qquad 
  \widehat{\mu^{TM}} = \mu^{TM}(\varphi) \left|\frac{\partial{\alpha}^*}{\partial s}\right| \,,
\end{equation}
where $\mu^{SM}(\varphi) = \mu^{TM}(\varphi) = 0$ for $\varphi < \overline\varphi<\pi$. For microscopic details of the model derivation see \cite{Manhart}.

Boundary conditions describe the forces acting on the filaments at their barbed ends and at the artificially introduced ends at the boundary of the
modeling domain:
\begin{eqnarray}
  &&\mu^B \partial_s\left(\eta \partial_s^2 F\right)  - \eta \lambda_\text{inext} \partial_s F 
  \pm \widehat{\mu^T}\eta\eta^* (\varphi-\varphi_0)\partial_s F^{\perp}
     \pm \widehat{\mu^{TM}}\eta\eta^* (\varphi-\pi)\partial_s F^{\perp} = -f_0 \,,\nonumber\\
  &&  \partial_s^2 F = 0 \,,\qquad\mbox{for } s=0 \,.\nonumber\\
  &&\mu^B \partial_s\left(\eta \partial_s^2 F\right)  - \eta \lambda_\text{inext} \partial_s F 
  \pm \widehat{\mu^T}\eta\eta^* (\varphi-\varphi_0)\partial_s F^{\perp}
     \pm \widehat{\mu^{TM}}\eta\eta^* (\varphi-\pi)\partial_s F^{\perp} = f_L \,,\nonumber\\
  &&  \partial_s^2 F = 0 \,,\qquad\mbox{for } s=-L \,.\label{BC}
\end{eqnarray}
Thus, there are no torques applied at the ends. The choice of the linear forces $f_0$ and $f_L$ along the leading edge and, respectively, along the 
artificial boundary will be discussed later. 

\section{Rigid actin filaments in the limit of large bending stiffness}\label{sec:rigid}

We want to derive a simplified model with rigid actin filaments. This is motivated on the one hand by the observation that filaments within the lamellipodium are typically rather straight \cite{Vinzenz}. On the other hand stiff filaments can be interpreted as a description of only the outermost part of the lamellipodial region, where filaments are (locally) straight. The resulting model is 
mathematically much simpler and can be derived by assuming a relatively large bending stiffness $\mu^B$. The limit
$\mu^B\to\infty$ will be carried out formally in this section.

The solutions of the formal limit
$$
  0 = \partial_s^2\left(\eta \partial_s^2 F\right)
$$
of \eqref{FBLMplus}, together with the boundary conditions
$$
   \partial_s^2 F = 0 \,,\qquad\mbox{for } s=0,-L \,,
$$
and with the constraint \eqref{inextensibility}, can be written as
\begin{equation} \label{rigid}
  F(\alpha,s,t) = F_0(\alpha,t) + (s-s_0(\alpha,t))d(\omega(\alpha,t)) \,,\qquad\mbox{with } d(\omega) = \binom{\cos\omega}{\sin\omega} \,,
\end{equation}
where $s_0$ is determined by 
$$
  \int_{-L}^0 \eta(\alpha,s,t)(s-s_0(\alpha,t))ds = 0 \,.
$$
In other words, $F_0$ is the center of mass of the filament, and $d(\omega)$ its direction. The components of $F_0$ and the angle $\omega$
are still to be determined. The total force balance obtained by integration of \eqref{FBLMplus} with
respect to $s$ and using the boundary conditions \eqref{BC} reads
\begin{align}
    f_0 + f_L = \int_{-L}^0 &\Bigl( \mu^A \eta D_t F+ \widehat{\mu^S}\eta\eta^* (D_t F - D_t^* F^*) \nonumber \\
  & + \widehat{\mu^{SM}}\eta\eta^* (D_t F - D_t^* F^*+ v^M(\partial_s F - \partial_s F^*))\Bigr) ds \,. \label{force-balance}
\end{align}
Note that it does not contain $\mu^B$ and therefore remains valid in the limit. Similarly, the total torque balance is obtained by
integration of \eqref{FBLMplus} against $(F-F_0)^\bot$:
\begin{align}
   (F-&F_0)^\bot(s=0)\cdot f_0 + (F-F_0)^\bot(s=-L)\cdot f_L \nonumber\\ 
  =& \mp \int_{-L}^0 \widehat{\mu^T} \eta\eta^*(\varphi-\varphi_0)ds
  \mp \int_{-L}^0 \widehat{\mu^{TM}} \eta\eta^*(\varphi-\pi)ds \nonumber\\
  &+ \int_{-L}^0 (F-F_0)^\bot\cdot \Bigl( \mu^A \eta D_t F+ \widehat{\mu^S}\eta\eta^* (D_t F - D_t^* F^*) \nonumber \\
  & + \widehat{\mu^{SM}}\eta\eta^* (D_t F - D_t^* F^*+ v^M(\partial_s F - \partial_s F^*))\Bigr) ds \,. \label{torque-balance}
\end{align}
This completes the formulation of the rigid filament version of the FBLM. Substitution of \eqref{rigid} into \eqref{force-balance} and \eqref{torque-balance}
gives a system of ordinary differential equations for $F_0$ and $\omega$. Note that coupling with respect to $\alpha$ happens only indirectly through
the interaction between the two filament families.

\section{A geometric simplification: the planar lamellipodium}\label{sec:planar}

Since in keratocytes the leading edge is rather smooth, we approximate a piece of lamellipodium by an infinite strip, parallel to the $x$-axis, and invariant to translations and to reflection. 
For the given data this means that the 
maximal filament length $L$ and the polymerization speed $v$ are constants. As a further simplification, we assume no filament ends inside the 
modeled part of the lamellipodium with the consequence $\eta = 1$ (and $s_0=-L/2$). 

We assume two families of rigid filaments \eqref{rigid} with
\begin{eqnarray*}
  F_0^+(\alpha^+,t) &=& \binom{x(t)+\alpha^+}{y(t)} \,,\quad \alpha^+\in\mathbb{R}\,,\qquad \omega^+(\alpha^+,t) = \omega(t) \in [0,\pi/2] \,,\\
  F_0^-(\alpha^-,t) &=& \binom{-x(t)+\alpha^-}{y(t)} \,,\quad \alpha^-\in\mathbb{R}\,,\qquad \omega^-(\alpha^-,t) = \pi-\omega(t) \in [\pi/2,\pi] \,,
\end{eqnarray*}
giving
$$
  F^\pm(\alpha^\pm,s^\pm,t)  = \binom{\pm x(t) + \alpha^\pm \pm (s^\pm+L/2)\cos\omega(t)}{y(t) + (s^\pm + L/2)\sin\omega(t)} \,,\qquad 
  \alpha^\pm\in\mathbb{R} \,,\quad s^\pm\in [-L,0] \,.
$$
The angle between two crossing filaments and the coordinate change between the two families mentioned in Section \ref{sec:add-myo} are easily computed:
$$
  \varphi=\pi-2\omega, \qquad \alpha^- = \alpha^+ + 2x(t) + (2s^+ + L)\cos\omega(t) \,,\qquad s^- = s^+ \,.
$$
It provides the geometric quantity needed in \eqref{coeff-cl} and \eqref{coeff-myo}: 
$$
  \left| \frac{\partial \alpha^-}{\partial s^+}\right| = 2\cos\omega \,.
$$
This quantity can be interpreted as a measure of the density of crossings, with a maximum at $\omega=0$ (fully collapsed lamellipodium) and a minimum at $\omega=\pi/2$ (all filaments are parallel, no crossings).

With the planar lamellipodium ansatz, the equations \eqref{force-balance} and \eqref{torque-balance} become independent of $\alpha$ and constitute a system
of three ordinary differential equations for the unknowns $(x(t),y(t),\omega(t))$:
\begin{eqnarray}
 \dot x \left[ \mu^A + 4(\mu^S + \mu^{SM}(\pi-2\omega))\cos\omega\right] &=& \frac{f_{0,x} + f_{L,x}}{L} + \mu^A v \cos\omega 
  + 4\mu^S v \cos^2\omega \nonumber\\
  &&+ 4\mu^{SM}(\pi-2\omega)(v-v^M)\cos^2\omega \,,\label{eq:x}\\
 \dot y \mu^A &=& \frac{f_{0,y} + f_{L,y}}{L} +\mu^A v \sin\omega \,,\label{eq:y}\\
 \dot\omega \left[ \mu^A + 4\sin^2\omega \cos\omega (\mu^S + \mu^{SM}(\pi-2\omega))\right] &=&  \frac{6}{L^2} d(\omega)^\bot \cdot (f_0-f_L)
      \nonumber \\
 && + \frac{24}{L^2} \mu^T (\pi-2\omega-\varphi_0)\cos\omega \nonumber \\
 && - \frac{48}{L^2} \mu^{TM}(\pi-2\omega)\omega \cos\omega \,.\label{eq:omega}
\end{eqnarray}

\section{Forces at the filament ends -- steady protrusion}\label{sec:steady}

The membrane stretched around the lamellipodium exerts a force on the polymerizing barbed ends. On the other hand, we assume that the
filaments at the rear of the lamellipodium are connected to stress fibres pulling them backwards, another consequence of actin-myosin interaction. Both the membrane force and the stress fibre force will be described as acting in the negative $y$-direction orthogonal to the leading edge, i.e. 
\begin{equation} \label{forces}
  f_{0,x} = f_{L,x} = 0 \,,\qquad f_{0,y} = -f_{mem}\,,\qquad f_{L,y} = -f_{stress} \,.
\end{equation}
If these forces are modeled as constant, the equation \eqref{eq:omega} for the angle is decoupled from the remaining system. For an analysis of its
dynamic behavior, we choose a model for the stiffness coefficients of the actin-myosin connection:
$$
  \mu^{SM}(\varphi) = \overline{\mu^{SM}} \,(\varphi - \overline\varphi)_+ \,,\qquad
  \mu^{TM}(\varphi) = \overline{\mu^{TM}} \,(\varphi - \overline\varphi)_+ \,,
$$
with $\overline{\mu^{SM}},\overline{\mu^{TM}} >0$, $\varphi_0 < \overline\varphi < \pi$, and with the notation $(.)_+$ for the positive part.

Bistability can now be obtained with appropriate assumptions on the parameters. The right hand side of \eqref{eq:omega} can be written as
$$
  \frac{24}{L^2} \cos\omega \left( \frac{f_{stress} - f_{mem}}{4} + h(\omega)\right) \qquad\mbox{with }
  h(\omega) = \mu^T(\pi-2\omega-\varphi_0) - 2\omega\overline{\mu^{TM}} (\pi-2\omega-\overline\varphi)_+ \,.
$$
It is a simple exercise to prove:

\begin{lem}\label{lem:h}
If
$$
\frac{\overline{\mu^{TM}}}{\mu^T}>\frac{\overline \varphi +\pi-2\varphi_0+2\sqrt{(\pi-\varphi_0)(\overline \varphi-\varphi_0)}}{(\pi-\overline \varphi)^2},
$$
then $h(\omega)$ as defined above has three simple zeroes $\omega_{10}$, $\omega_{20}$, $\omega_{30}$, satisfying
$$
  \frac{\pi}{2} > \omega_{10} = \frac{\pi  -\varphi_0}{2} > \frac{\pi-\overline\varphi}{2} > \omega_{20} > \omega_{30} > 0 \,.
$$
\end{lem}

\begin{thm}
Under the assumptions of Lemma \ref{lem:h} and for $|f_{stress}-f_{mem}|$ small enough, the ordinary differential equation \eqref{eq:omega} with the forces given by
\eqref{forces} possesses four stationary solutions $\omega_j$, $j=0,\ldots,3$ with 
$$
  \omega_0=\pi/2 > \omega_1 = \frac{\pi  -\varphi_0}{2} + \frac{f_{stress} - f_{mem}}{8\mu^T} > \frac{\pi-\overline\varphi}{2} 
  > \omega_2 > \omega_3 > 0 \,,
$$
where $\omega_0$ and $\omega_2$ are unstable, and $\omega_1$ and $\omega_3$ are asymptotically stable.
\end{thm}

Again the proof is straightforward.
For the stable steady states, the lamellipodium has the constant protrusion speeds
$$
  \dot y = v \sin\omega_{1,3} - \frac{f_{stress} + f_{mem}}{\mu^A L} \,.
$$
For the equilibrium angle $\omega_1$, we typically expect the speed to be positive. 
It is not affected by actin-myosin interaction. The smaller speed
corresponding to $\omega_3$ might actually be negative due to membrane tension and stress fibres, i.e. the second stable state, where the lamellipodium 
is collapsed by actin-myosin interaction, might be retractive.

Finally, the steady states also produce lateral flow with constant speeds
$$
  \dot x = v\cos\omega_1 \qquad\mbox{and}\qquad 
  \dot x = \left( v - v^M \frac{4\mu^{SM}\cos\omega_3}{\mu^A + 4\mu^S\cos\omega_3 + 4\mu^{SM}\cos\omega_3}\right)\cos\omega_3 \,,
$$
respectively, where in the collapsed state the lateral flow speed produced by polymerization is reduced by actin-myosin interaction.

\section{Coupling of two opposing lamellipodia}\label{sec:coupling}

As a caricature of a cell fragment, we consider two back-to-back planar lamellipodia (see Figure \ref{fig:fig1}E). For notational convenience, the bottom lamellipodium is rotated by $180^o$ in the mathematical description. Therefore we consider two versions of the system
\eqref{eq:x}--\eqref{eq:omega} with unknowns $(x,y,\omega)$ and $(\hat x, \hat y, \hat\omega)$. The assumption that the total forces exerted on the
fragment by membrane tension and by stress fibres vanish, imply that \eqref{forces} is used in both systems with the same values for $f_{mem}$ and
$f_{stress}$. However, we allow the option that these forces are not constant but regulate the size of the fragment, measured by $y+\hat y$.
We first consider the case of a constant given membrane force and a size dependent force by stress fibres:
$$
  \mbox{\bf Case A:} \qquad f_{mem} =\mbox{const} \,, \quad f_{stress} = f_{stress}(y+\hat y) \,.
$$
Typically $f_{stress}$ will be an increasing function, but the details are not important for our considerations.

Adding the equations \eqref{eq:y} for $y$ and $\hat y$ leads to a closed system of three equations for $y+\hat y$, $\omega$, and $\hat\omega$:
\begin{eqnarray}
 (\dot y + \dot{\hat y}) \mu^A &=& -\frac{2}{L}(f_{mem} + f_{stress}(y+\hat y)) +\mu^A v (\sin\omega + \sin\hat\omega) \,,\label{eq:y+yhat}\\
 \dot\omega \,g(\omega) &=&   \cos\omega \left( \frac{f_{stress}(y+\hat y) - f_{mem}}{4} + h(\omega)\right) \,,\label{eq:om1} \\
  \dot{\hat\omega} \,g(\hat\omega) &=&   \cos\hat\omega \left( \frac{f_{stress}(y+\hat y) - f_{mem}}{4} + h(\hat\omega)\right) \,,
      \label{eq:om2}
\end{eqnarray}
with
$$
  g(\omega) = \frac{L^2}{24} \left[ \mu^A + 4\sin^2\omega \cos\omega (\mu^S + \mu^{SM}(\pi-2\omega))\right] \,.
$$
We shall prove that with appropriate assumptions on the data, the problem has 4 stable steady states.

\begin{thm} \label{thm:equilibria}
Let the assumptions of Lemma \ref{lem:h} hold, let the function $f_{stress}$ be continuously differentiable with bounded positive derivative, and let $f_{mem}$, 
$\mu^A vL$, and the Lipschitz constant of $f_{stress}$ be small enough. Then the system \eqref{eq:y+yhat}--\eqref{eq:om2} has four stable
steady states, satisfying
\begin{eqnarray}
&&\omega=\hat\omega = \omega_{10} + O(f_{mem} + \mu^A vL) \,, \label{eqn:stat1}\\
&&\omega=\hat\omega = \omega_{30} + O(f_{mem} + \mu^A vL) \,,\label{eqn:stat2}\\
&&\omega=\omega_{10} + O(f_{mem} + \mu^A vL) \,,\quad \hat\omega = \omega_{30} + O(f_{mem} + \mu^A vL) \,,\label{eqn:mov1}\\
&&\omega=\omega_{30} + O(f_{mem} + \mu^A vL) \,,\quad \hat\omega = \omega_{10} + O(f_{mem} + \mu^A vL) \,.\label{eqn:mov2}
\end{eqnarray}
\end{thm}

\begin{proof}
From \eqref{eq:y+yhat} we obtain that steady states have to satisfy
$$
  f_{stress}(y+\hat y) = -f_{mem} + \frac{\mu^A Lv}{2}(\sin\omega + \sin\hat\omega) \,.
$$
This implies, again for stable steady states, $h(\omega) = h(\hat\omega) = O(f_{mem} + \mu^A vL)$. The existence of the four steady states is then a 
consequence of a straightforward perturbation argument. The coefficient matrix in the linearization of \eqref{eq:y+yhat}--\eqref{eq:om2} can be written as
$$
 \left( \begin{array}{ccc}  -2\kappa/(\mu^A L)&v\cos\omega & v\cos\hat\omega \\
                                       A\kappa& Ah^\prime(\omega)& 0 \\
                                        \hat A \kappa& 0&\hat A h^\prime(\hat\omega)\end{array}\right) \,,
$$
with positive constants $A$ and $\hat A$, and with $0<\kappa = f_{stress}^\prime(y+\hat y) \ll 1$. A perturbation analysis of the eigenvalue problem
for small $\kappa$ (i.e. formal expansion of eigenvalues in terms of powers of $\kappa$
and subsequent justification by a contraction argument) gives the eigenvalues
$$
  \lambda_1 = Ah^\prime(\omega) + O(\kappa) \,,\quad \lambda_2 = \hat Ah^\prime(\hat\omega) +O(\kappa) \,,\quad
  \lambda_3 = \kappa\left( -\frac{2}{\mu^A L} + \frac{v\cos\omega}{h^\prime(\omega)} + \frac{v\cos\hat\omega}{h^\prime(\hat\omega)}\right) 
  + O(\kappa^2) \,,
$$
which are all negative at the four steady states for small enough $\kappa$, because of $h^\prime(\omega_{10}), h^\prime(\omega_{30}) < 0$.
\end{proof}

For the steady states the protrusion speed of the fragment is constant and given by
$$
  \dot y = -\dot{\hat y} = \frac{v}{2}(\sin\omega - \sin\hat\omega) \,.
$$
For the symmetric steady states \eqref{eqn:stat1}, \eqref{eqn:stat2}, the protrusion speeds vanish, hence they describe stationary cells (or fragments). The equilibrium angles in the lamellipodia in this case are either both affected by myosin, \eqref{eqn:stat2}, or both result only from cross-link activity, \eqref{eqn:stat1}. The asymmetric steady states \eqref{eqn:mov1} and \eqref{eqn:mov2} describe a protruding, polarized cell. In both cases it consists of a collapsed cell rear, in which myosin is active ($\omega=\omega_{30}$), and a cell front with a steeper equilibrium angle caused only by cross-link activity ($\omega=\omega_{10}$).

Finally, we also mention the case of a constant stress fibre force and a size dependent membrane force:
$$
  \mbox{\bf Case B:} \qquad f_{stress} =\mbox{const} \,, \quad f_{mem} = f_{mem}(y+\hat y) \,.
$$
Without going through the details, we note that the qualitative results are the same and a theorem analogous to Theorem \ref{thm:equilibria} can be proven.

\section{Parameter values -- simulations with the full model}\label{sec:sim}

In this section we demonstrate that with the additional term describing myosin within the lamellipodium, the model is able to produce cells/cell fragments that, depending on the initial conditions, will either remain stationary or start moving. In contrast to the simulations presented in \cite{MOSS} and \cite{MOSS2}, here the movement is achieved without an external cue and without varying the polymerization speed. In the simulation, we work with the full model \eqref{FBLMplus}--\eqref{BC} and not with the simplifications introduced in Sections \ref{sec:rigid} and \ref{sec:planar}. However, the qualitative results of Section \ref{sec:coupling} will be reproduced.

\paragraph{Parameter values:}
Parameter values are chosen as in \cite{MOSS} with the following exceptions and additions:
we work with a constant filament density $\eta=1$ in parameter space, which means that the filament number remains constant with branching and capping always in equilibrium. No pointed ends appear within the simulation region, which 
corresponds to a fixed filament length of $L=3\mu m$. The polymerization speed is fixed at the constant value 
$v=1.5 \mu m\, \text{min}^{-1}$. In \cite{Svitkina} it has been observed that myosin speckles that are formed in the lamellipodium drift inwards with time. This indicates that the myosin velocity has to be smaller than the polymerization speed. We therefore chose $v^M=1 \mu m\, \text{min}^{-1}$. We assume that myosin can only act on actin filaments if the angle between the filaments is more than $\overline \varphi=120\degree$. For the stiffness parameters of stretching and twisting the cross-links and myosin good estimates are hard to obtain, since their exact concentration in the lamellipodium is difficult to determine. However, motivated by Lemma \ref{lem:h}, which requires that $\frac{\overline{\mu^{TM}}}{\mu^T}> 4.91$ we chose that ratio to be 5. For the membrane and stress-fiber forces we chose $f_\text{mem}=\mu^\text{mem}\cdot(d_\text{out}-\overline{d_\text{out}})_+$ and $f_\text{stress}=\mu^\text{stress}\cdot (d_\text{in}-\overline{d_\text{in}})_+$, where $d_\text{out}$ is the fragments' averaged 
outer diameter calculated from the area $A_\text{out}$ by $d_\text{out}=2\sqrt{\frac{A_\text{out}}{\pi}}$. If the total area is replaced by the inner area of the cell without the lamellipodium, one correspondingly obtains the expression for $d_\text{in}$. Additionally we increase the bending stiffness by a factor 10 in order to get closer to the analytical case examined in Sections \ref{sec:rigid}.

\begin{figure}[h!]
        \centering
	\includegraphics[width=\textwidth]{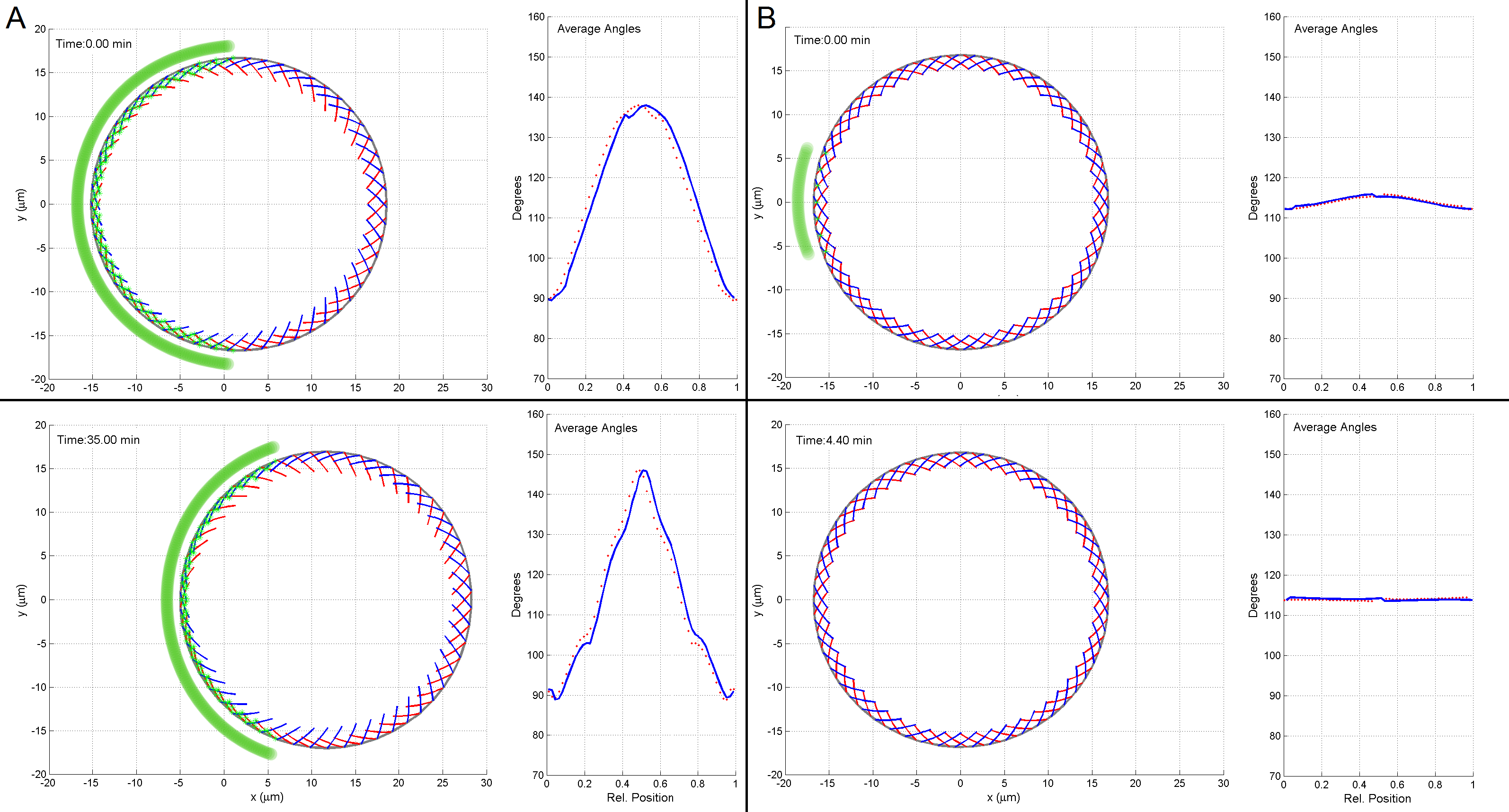}
	\caption{\small Cell view with clockwise filaments in blue and anti-clockwise
	filaments in red. Green stars in the lamellipodium mark the area where myosin is active, further emphasized by the green bands. On the right of each fragment: angle between the filaments of the different families, averaged along the filaments, parametrized along the membrane with 0 being at the very right and going counterclockwise. Top Row: Initial conditions leading to A: a moving fragment, B: a stationary fragment. Bottom Row: Filament positions and average angles at a later time after equilibrium has been reached. Parameters as in Table \ref{tab:parameters}.}
	\label{fig:sim}
\end{figure}

\paragraph{Simulation results:}
Figure \ref{fig:sim} shows the initial conditions and steady state situation for two different simulations done with the same parameters. On the left (Figure \ref{fig:sim}A) a cell is shown, where due to rather anti-parallel angles, initially myosin is able to act within about half the fragment. In this situation an equilibrium is attained in which in the right half of the cell no myosin is active and the angles between filaments are rather steep. In the left half the equlibrium angles attained in the presence of myosin are more anti-parallel. This leads to an equilibrium state in which the fragment moves steadily to the right (see movie in Supplementary Material), which corresponds to a situation described by the steady states \eqref{eqn:mov1} and \eqref{eqn:mov2} in Theorem \ref{thm:equilibria}. On the right (Figure \ref{fig:sim}B), initial conditions have been used where only in a small area of the leftmost part of the fragment, myosin can act on the filaments. However this is not enough to establish 
itself there permanently and hence after a short time the fragment reverts to its rotationally symmetric form and remains stationary. This situation corresponds to the steady state \eqref{eqn:stat1} in Theorem \ref{thm:equilibria}.

\begin{table}[h!]
\caption{Parameter Values}
\label{tab:parameters}       

\begin{tabular}{| c | p{4cm} | p{4cm} | p{4cm} | }
\hline
Var. & Meaning & Value & Comment \\ 

\hline
& & & \\
$\mu^B$ & bending elasticity & $0.7pN\mu m^2$ & 10 times higher than in \cite{Gittes}  \\
$\mu^A$ & macroscopic friction caused by adhesions & $0.14 pN \text{min} \mu m^{-2}$ & measurements in \cite{Li,Oberhauser}, estimation and calculations in \cite{OelSch3,OelSch2,OelSch1}\\
$\overline{d_\text{in}}$ & equilibrium inner diameter & $27.6 \mu m$ & order of magnitude as in \cite{Verkhovsky}\\
$\overline{d_\text{out}}$ & equilibrium outer diameter & $33.3 \mu m$ & order of magnitude as in \cite{Verkhovsky}\\
$v$ &  polymerization speed & $1.5\mu m \,\text{min}^{-1}$ & in biological range\\
$\varphi_0$ & equilibrium cross-link angle & 70\degree & equal to the branching angle\\
$\mu^S$ & cross-link stretching constant & $2.6\!\times\! 10^{-2} pN\, min\, \mu m^{-1} $& \\
$\mu^T$ & cross-link twisting constant & $0.21 pN\,\mu m$ & \cite{OelSch3, Goldmann} and computations in \cite{OelSch2,OelSch1}\\
$v^M$ & myosin velocity & $1 \mu m\,\text{min}^{-1}$ & order of magnitudes as in \cite{Svitkina}\\
$\overline \varphi$ & myosin cut-off & $120\degree$ &\\
$\overline{\mu^{SM}}$ & myosin stretching constant & $2.6\!\times\! 10^{-2} pN\, \text{min}\, \mu m^{-1} $& \\
$\overline{\mu^{TM}}$ & myosin twisting constant & $1 pN\,\mu m$ & motivated by Lemma \ref{lem:h}\\
$\mu^\text{mem}$ & membrane force & $5\!\times\! 10^{-3} pN \mu m^{-1}  $ &\\
$\mu^\text{stress}$ & stress fiber force & $5\!\times\! 10^{-2}  pN \mu m^{-1}  $ &\\

\hline
\end{tabular}
\end{table}

\end {document}